\documentclass[twocolumn]{svjour3}          
\smartqed  
\usepackage{graphicx}
\usepackage{amsmath}
\usepackage{enumitem}

\begin{document}

\title{A Novel Observer Design for LuGre Friction Estimation and Control}


\author{Caner Odabaş         \and
        Ömer Morgül 
}


\institute{C. Odabaş \at
              Aselsan Radar, Electronic Warfare and Intelligence Systems Division, 06830, Ankara, Turkey \\
              \email{codabas@aselsan.com}           
           \and
           Ö. Morgül \at
              Department of Electrical and Electronics Engineering, Bilkent University, 06800, Ankara, Turkey
              \email{morgul@ee.bilkent.edu.tr}  
}


\maketitle

\begin{abstract}
    Dynamic components of the friction may directly impact the stability and performance of the motion control systems. The LuGre model is a prevalent friction model utilized to express this dynamic behavior. Since the LuGre model is very comprehensive, friction compensation based on it might be challenging. Inspired by this, we develop a novel observer to estimate and compensate for LuGre friction. Furthermore, we present a Lyapunov stability analysis to show that observer dynamics are asymptotically stable under certain conditions. Compared to its counterparts, the proposed observer constitutes a simple and standalone scheme that can be utilized with arbitrary control inputs in a straightforward way. As a primary difference, the presented observer estimates velocity and uses the velocity error to estimate friction in addition to control input. The extensive simulations revealed that the introduced observer enhances position and velocity tracking performance in the presence of friction. 
\end{abstract}

\section{Introduction}\label{Intro}
According to da Vinci, friction is an opposing force to motion. Also, it is proportional to the normal force and independent of the contact area. However, Coulomb proposed the first mathematical representation. In the reputed Coulomb model, friction is also independent of the velocity magnitude. In reality, friction is a more complex nonlinearity having multi-stage
characteristics. For instance, friction force at rest is higher than the Coulomb friction level and this characteristic is called stiction. In this phase, even if the system preserves its position, a microscopic presliding motion occurs. In this case, the minimum force required to
generate movement is defined as the break-away force. A more realistic and continuous friction model includes the Stribeck effect. This component is dominant at low velocities and
becomes minimal at a certain velocity called Stribeck velocity. In detail, friction decreases continuously starting from the stiction level while velocity increases. Likewise, the lubricity between two contacting surfaces influences friction behavior. This phenomenon is proportional to the velocity and defines viscous friction. In brief, a classical friction model can be acquired by integrating the stiction, viscous coefficient, and Stribeck effect to the Coulomb model. Nevertheless, the classical model represents the steady-state friction characteristics well since it portrays only the sliding phase. On the contrary, the dynamic models, which include some time-varying parameters, can express the pre-sliding phase of motion to a better extent.
Hence, many dynamic friction models have been developed in the literature. Although the comparison and details of these models are out of the scope of this paper, interested readers may refer to \cite{Marques_ND.2016} and the references therein. This study focuses on a benchmark dynamic friction model presented by researchers from the universities of Lund and Grenoble \cite{Canudas_IEEETAC.1995}. According to proposed LuGre model, when a
tangential force is applied, the elastic bristles connecting the moving object and asperities start to bend like a spring-mass system. If the force is sufficiently large, deflected bristles slip off each other randomly. New contacts are formed and slipped off repetitively throughout the motion, generating friction. As a result, the LuGre model promisingly expresses position-dependent pre-sliding behaviors such as friction lag, break-away force or stick-slip motion.

In addition to modeling, friction identification is another comprehensive and crucial research area in friction compensation since the desired system response may only be achieved when model parameters coincide with the existing friction coefficients. Sometimes, the identification procedure might be very demanding and require diverse experimental analyses or exhaustive setups. Moreover, friction parameters might vary in time due to environment, temperature, material properties, position, etc. Therefore, friction observer designs are prevalent in the literature to improve the performance of systems. While some observer estimations require elementary identification processes, some exhibit a fully adaptive nature. For instance, simulations and experiments revealed that although the Friedland-Park observer is designed to aim at Coulomb friction, it enhances the tracking response \cite{Friedland_IEEETAC.1992} even if the actual friction exhibits nonlinear characteristics which cannot be confined to the Coulomb Model. Based on some assumptions, \cite{Tafazoli_IEEETCST.1998} provides a guideline to determine design parameters in Friedland-Park Observer and makes a velocity estimation in addition to friction. Also, \cite{caner_tez.2014} extends the initial Friedland-Park Observer to a general structure and establishes the necessary and sufficient conditions for the adaptation function. \cite{caner_makale.2014} modifies the original Friedland-Park observer for systems with time delay. Similarly,  \cite{Ruderman_IEEETIE.2014} devises a feed-forward friction observer for the motor drive plant with a non-negligible time delay under a two-state dynamic friction model with elastoplasticity. For the same dynamic friction model, \cite{Ruderman_IEEETIE.2015} investigates the viscous and Coulomb friction uncertainties and derives an exponentially stable observer discussing the minimum necessary system identification. In another approach, \cite{Freidovich_IEEETCST.2010} considers a feedback system designed without disturbance concerns. Assuming that the LuGre model adequately exhibits the disturbance characteristics within the system, \cite{Freidovich_IEEETCST.2010} introduces a Lyapunov-function-based LuGre friction observer. Alternatively, for a robotic manipulator with LuGre friction, \cite{xu_IJCAS.2022} implements a backstepping sliding mode control based on an extended state observer. In this case, the Lyapunov criterion is applied to obtain a backstepping method satisfying global asymptotic stability. Lastly, \cite{chen_IJRNC.2022} entirely focuses on the controller design for mechanical systems with LuGre friction. To this end, they exhibit an adaptive repetitive learning control method for periodic reference signals and an adaptive dynamic state feedback control for constant reference inputs when all coefficients of the LuGre model are unknown. In this paper, we exhibit a novel standalone LuGre friction observer by considering a basic mechanical structure with LuGre friction and show that under certain conditions the error dynamics is asymptotically or exponentially stable.  

This paper is organized as follows. In the next section, we share the definition of the problem and explain the motivation behind this study. Then, we briefly mention a prominent LuGre friction observer structure and introduce a novel design. Also, we discuss the contributions of the presented LuGre friction observer and provide a detailed stability analysis. In Section 4, we perform simulations to demonstrate the effectiveness of the proposed friction compensation scheme under different scenarios and then give some concluding remarks.

\section{Problem Formulation}
Motivated by classical mechanical structures, we consider the following system:
\begin{equation}
\label{mechanical_sys1}
J\dot{w}=-F+u,
\end{equation}
where $J$ represents inertia, $w$ represents the (linear or angular) velocity, $F$ represents the friction term and $u$ is the applied input (force or torque). If $\hat{F}$ is an estimation of the actual friction term $F$, one natural control methodology would be to choose the control input as $u=\hat{F}+v$ where $v$ is the new control input. If we have $\hat{F}=F$, then (\ref{mechanical_sys1}) reduces to:
\begin{equation}
\label{sys1}
J\dot{w}=v.
\end{equation} 
Obviously, the system given by (\ref{sys1}) can be controlled by standard techniques to achieve various objectives such as stabilization, tracking, etc. Nevertheless, if $\hat{F} \neq F$, we obtain
\begin{equation}
\label{mechanical_sys2}
J\dot{w}=-F+\hat{F}+v.
\end{equation} 

Depending on the properties of the friction estimation $\hat{F}$, the control law $v$ used in an ideal (frictionless) case given by (\ref{sys1}) to solve some particular control problem could also be used in (\ref{mechanical_sys2}) to solve the same problem. We will utilize this approach in our work, which could be called a controller design based on friction compensation.
The estimation $\hat{F}$ naturally depends on the friction model $F$ utilized in (\ref{mechanical_sys1}). Especially in this work, we will consider the LuGre model given:
\begin{eqnarray}
\label{LuGre1}
\dot{z} &=&w-\sigma_{0}\frac{|w|}{h(w)}z,\\
F&=&\sigma_{0}z+\sigma_{1}\dot{z}+F_{v}w. \label{LuGre2} 	
\end{eqnarray} 
Here, $z$ is a non-measurable internal state describing the bristles' relative deflection during the stiction phases  \cite{canudas_IFAC.1996}. Also, $\sigma_{0}$ and $\sigma_{1}$ can be considered as the stiffness coefficient of the microscopic deformations of $z$ and the damping coefficient associated with $\dot{z}$. A particular choice of $h(w)$ is given as 
\begin{equation}
\label{Lugre3}
h(w)=F_{c}+(F_{s}-F_{c})e^{-\left( \frac{w}{w_{s}}\right) ^{2}},
\end{equation} 
where $F_{c}$, $F_{s}$ and $w_{s}$ represent Coulomb friction, stiction and Stribeck velocity, respectively. Without loss of generality, we assume that the viscous friction coefficient, $F_{v}$ is exactly known. Note that if $\dot{z} \to 0$, then $z \to h(w)sgn(w)/\sigma_{0}$. This leads to the classical static friction model given as
\begin{equation}
\label{Lugre_ss}
F=F_{c}sgn(w)+(F_{s}-F_{c})e^{-\left( \frac{w}{w_{s}}\right) ^{2}}sgn(w)+F_{v}w. 
\end{equation} 
\begin{figure}[tb]
	\centering
	\includegraphics[width=1\columnwidth]{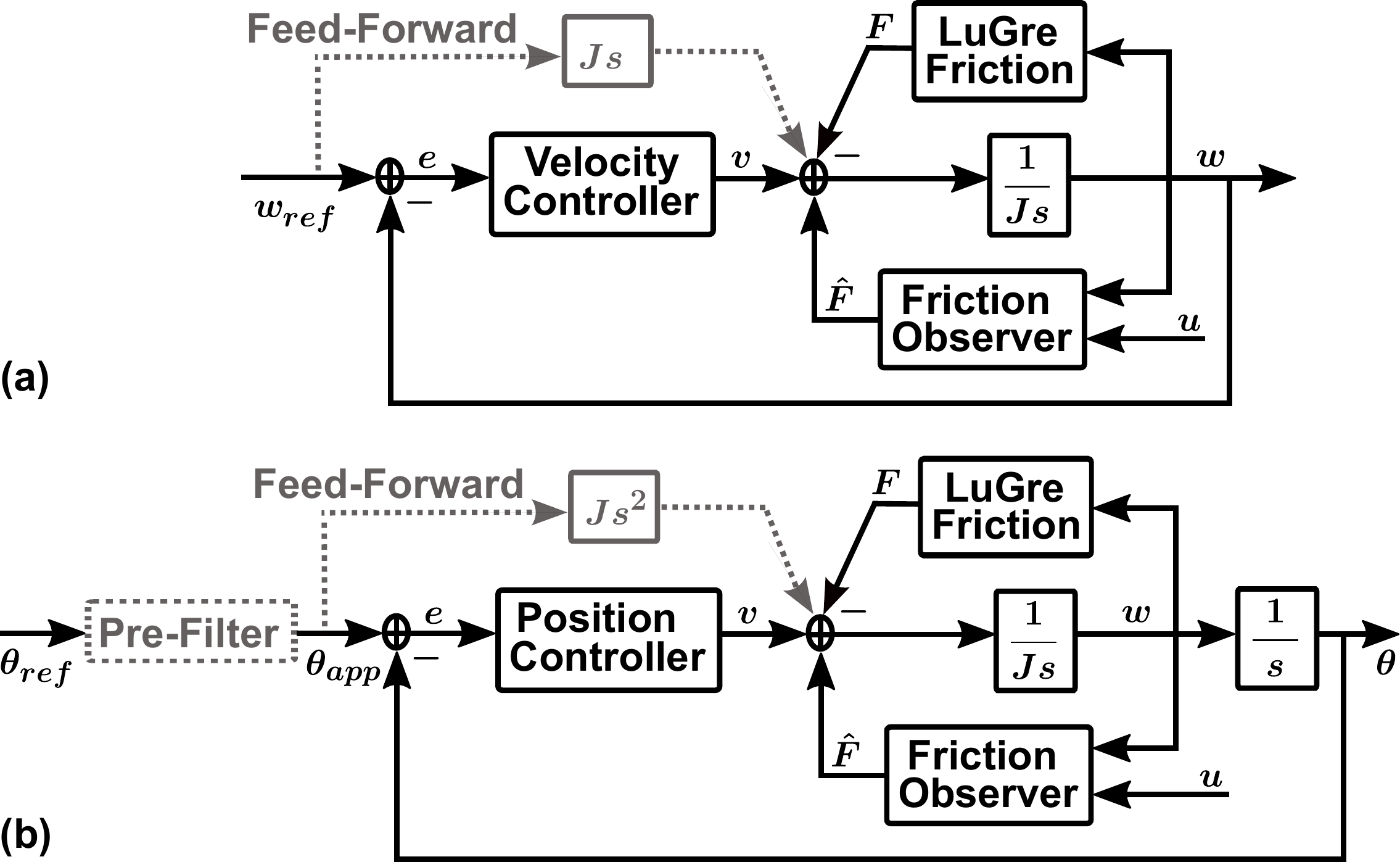}
	\caption{\label{sys} Block diagrams of feedback systems under under LuGre Friction. (a) Velocity control. (b) Position control. }
\end{figure}

To obtain a reasonable estimation for the friction term given by (\ref{LuGre2}), we need estimate $z$ first. Considering (\ref{LuGre1}), the state $z$ is not measurable while velocity $w$ is. Our aim is to obtain a state observer to estimate $z$ in (\ref{LuGre1}) and friction given by (\ref{LuGre2}). Then, we will discuss the effect of this estimation on the friction compensation-based controller design alluded to above. To summarize, our motivation is to propose a novel stand-alone LuGre friction observer design; therefore, we consider both velocity and position loop structure, as shown in Fig. 1. Certainly, the pre-filter and feed-forward terms, represented by dashed gray, are optional components that can enhance the tracking performance of the system. Although the introduced observer scheme does not require a feed-forward term and particular controller design constraints, in the literature, some observer structures do.

\section{Observer Structures}
\subsection{A Brief LuGre Friction Observer Survey}
Assume that for the system given by (\ref{mechanical_sys1}), (\ref{LuGre1})-(\ref{LuGre2}), $w(.)$ is measured. To estimate $z$ given by (\ref{LuGre1}), a natural choice is to use its exact copy. In this case, deflection and friction estimations are given as
\begin{eqnarray}
\label{z_hat_dot}
\dot{\hat{z}}&=&w-\sigma_{0}\frac{|w|}{h(w)}\hat{z}, \\
\label{F_est}
\hat{F}&=&\sigma_{0}\hat{z}+\sigma_{1}\dot{\hat{z}}+F_{v}w. 
\end{eqnarray}
This natural observer is first given in \cite{Gafvert_ECC.1999}. To analyze the properties of the natural observer given by (\ref{z_hat_dot}), let us define the deflection estimation error as $e_{z}=z-\hat{z}$.
By using (\ref{LuGre1}) and (\ref{z_hat_dot}), we obtain:
\begin{equation}
\label{err_z_dot}
\dot{e}_{z}=-\sigma_{0}\frac{|w|}{h(w)}e_{z}. 
\end{equation}
Similarly, friction estimation can be defined as $e_{f}=F-\hat{F}.$ By using (\ref{LuGre2}), (\ref{z_hat_dot}) and (\ref{F_est}), we obtain:
\begin{equation}
\label{F_err}
e_{f}=\sigma_{0}e_{z}+\sigma_{1}\dot{e}_{z}=\sigma_{0}\left(1-\frac{|w|}{h(w)} \right)e_{z}.
\end{equation}
Under reasonable conditions on $w(.)$, one can prove the asymptotic stability of the error dynamics in (\ref{err_z_dot}). To see that, let us define the Lyapunov function, $
V=\frac{1}{2}e^{2}_{z}.$ Then, by using (\ref{err_z_dot}), we obtain:
\begin{equation}
\label{V_dot}
\dot{V}=-\sigma_{0}\frac{|w|}{h(w)}e_{z}^{2}.
\end{equation}
This proves that $V$ is bounded, hence $e_{z} \in\mathbf{ L_{\infty}}$. However, since (\ref{err_z_dot}) is a time-varying system if $w(.)$ is a function of time, we cannot utilize LaSalle's Invariance Principle. Hence, asymptotic stability does not follow directly from (\ref{V_dot}). The stability properties of the error system given by (\ref{err_z_dot}) naturally depend on the properties of $w(.)$. Two such stability results that immediately follow from (\ref{err_z_dot}) are given below.
\begin{lemma}\label{lemma1}
	Assume that $| w(t) | \geq \alpha$ for some $\alpha>0$. Then, the error system in (\ref{err_z_dot}) is exponentially stable.
\end{lemma}
\begin{proof}\label{proof1}
	It follows from (\ref{Lugre3}) that
	\begin{equation}
	\label{proof1_eq}
	C_{1} \leq h(w) \leq C_{2},~~~~~~~~ \forall w.
	\end{equation}
	where $C_{1}=min\left\lbrace F_{s}, F_{c}\right\rbrace \geq 0$, $C_{2}=max\left\lbrace F_{s}, F_{c}\right\rbrace>0$. It then follows from (\ref{V_dot}) that
	\begin{equation}
	\label{proof1_eq2}
	\dot{V} \leq -\frac{\sigma_{0}\alpha}{C_{2}}e^{2}_{z}=-\frac{2\sigma_{0}\alpha}{C_{2}}V.
	\end{equation}
	Exponential stability follows from (\ref{proof1_eq2}). See also \cite{odabas_TIMC.2021,odabas_IJC.2023}
\end{proof}
\begin{lemma}\label{lemma2}
	Assume that the following is satisfied for some $T>0$ and $\beta>0.$
	\begin{equation}
	\label{lemma2_eq}
	\int_{t}^{t+T}|w(\tau)|d\tau \geq \beta,~~~~~~~~ \forall t \geq 0.
	\end{equation}
	Then, the error dynamics in (\ref{err_z_dot}) is exponentially stable.	
\end{lemma}
\begin{proof}\label{proof2}
	Note that the solution of (\ref{err_z_dot}) is given by the following:
	\begin{equation}
	\label{proof2_eq}
	e_{z}(t)=e^{-\int_{0}^{t}\sigma_{0}\frac{|w(\tau)|}{h(w(\tau))}d\tau}e_{z}(0),~~~~~~~~ \forall t \geq 0.
	\end{equation}
	It follows from (\ref{proof1_eq}) and (\ref{lemma2_eq}) that we have:
	\begin{equation}
	\label{proof2_eq2}
	\int_{t}^{t+T}\frac{|w(\tau)|}{h(w(\tau))}d\tau\geq\frac{\sigma_{0}}{C_{2}}\int_{t}^{t+T}|w(\tau)|d\tau\geq\frac{\sigma_{0}}{C_{2}}\beta>0.
	\end{equation}
	Then exponential stability easily follows from (\ref{proof2_eq}) and (\ref{proof2_eq2}). See e.g. \cite{odabas_TIMC.2021,odabas_IJC.2023} for similar results. 		
\end{proof}
\begin{remark}\label{remark1}
	Note that when $|w(t)|\geq\alpha$, (\ref{lemma2_eq}) is always satisfied. However, the converse is not true. Indeed, if $w(t)=sin(2\pi ft)$ for a frequency $f$, then (\ref{lemma2_eq}) is satisfied $T=\frac{1}{2f}$ and $\beta=\frac{1}{\pi f} $. Lemma \ref{lemma1} and \ref{lemma2} can be utilized in unit step and sinusoidal signal tracking respectively. 
\end{remark}
\begin{remark}\label{remark2}
	Clearly, (\ref{proof2_eq}) indicates that to have asymptotically stable error dynamics given by (\ref{err_z_dot}), the signal $w(.)$ should satisfy certain properties, as indicated in Lemma \ref{lemma1} and \ref{lemma2}. This is the reminiscent of ``persistency of excitation'' requirement frequently encountered in adaptive control theory, see e.g. \cite{sastry_adaptivebook.2011}. For example, if $w\in \mathbf{L_{1}}$, then $e_{z}(t)$ given by (\ref{proof2_eq}) converges to a finite limit; hence, error dynamics in (\ref{err_z_dot}) is not asymptotically stable. This shows that the observer given by (\ref{z_hat_dot}) is not suitable for arbitrary inputs $u$ in (\ref{mechanical_sys1}) or $v$ in (\ref{mechanical_sys2}).
\end{remark}

\begin{lemma}\label{lemma3}
	Assume that $w\in \mathbf{L_{\infty}}$. If the observer error dynamics given by (\ref{err_z_dot}) is asymptotically stable, then (\ref{F_est}) gives an asymptotically stable friction estimator, i.e. $\lim_{t \to \infty} \hat{F} = F$. If the observer error dynamics is exponentially stable, then (\ref{F_est}) gives an exponentially fast friction estimation, i.e. for some $C>0$ and $\delta>0$ $|F-\hat{F}| \leq Ce^{-\delta t}, \forall t \geq 0$ holds.
\end{lemma}
\begin{remark}\label{remark3}
	Lemma \ref{lemma3} is the basic rationale behind using the method of controller design on friction compensation. Indeed, if Lemma \ref{lemma3} holds, then the controller design based on (\ref{sys1}), i.e. the frictionless case, could also be used in (4) to solve the same control objective by treating the friction estimation error as a decaying disturbance. 
\end{remark}
Unfortunately, high-quality data measurement and parameter identification are vital for this observer structure. As previously explained, parameter identification might be challenging and these coefficients can vary due to normal force variations or environmental effects such as temperature change. Moreover, in an overcompensated case with a PID controller, the integral part deals with the friction estimation error until the zero crossing. Thus, the estimated friction force changes rapidly as the velocity changes sign, resulting in an extra accelerating torque and a control error \cite{Olsson_CDC.1996}. \cite{caner_tez.2021} shows that a large integral coefficient exhibits an oscillatory behavior on tracking response when static friction parameters $F_{c}$ and $F_{s}$ are over-identified. To overcome this problem, \cite{Canudas_IEEETAC.1995} incorporate controller error, $e$ into (\ref{z_hat_dot}) as a correcting factor. $e$ can be defined as the difference between the measured position and applied position reference, i.e., $e=\theta-\theta_{app}$ for the position tracking objective or similarly the difference between the measured velocity and applied velocity reference, i.e.,  $e=w-w_{ref}$ for the velocity tracking objective. In addition to this modification, control law should include a feed-forward term as given in order to satisfy asymptotically stable position/velocity and estimation dynamics. Therefore, for a position control problem, the following observer-controller structure is proposed in \cite{Canudas_IEEETAC.1995}:
\begin{eqnarray}
\label{Canudas_z}
\dot{\hat{z}}&=&w-\sigma_{0}\frac{|w|}{h(w)}\hat{z}-ke, \\
u&=&-C(s)e+\hat{F}+J\frac{d^{2}\theta_{app}}{dt^{2}}, 		\label{Canudas_u}
\end{eqnarray}
where $k>0$ denotes the observer gain and $C(s)$ stands for a position controller. Also, note that although the feed-forward component, $Js^2$ or $d^{2}\theta_{app}/dt^{2}$,  is an optional part in Fig. \ref{sys}, the introduced position control strategy in \cite{Canudas_IEEETAC.1995} requires the feed-forward term in order to obtain error dynamics converging to zero. The feed-forward component is also a crucial requirement to determine the controller. Also according to \cite{Canudas_IEEETAC.1995}, $C(s)$ must be chosen appropriately to ensure $T(s)=\frac{\sigma_{1}s+\sigma_{0}}{Js^2+C(s)}$ is a strictly positive real (SPR) transfer function. Then the observer and position error will be zero asymptotically. Later, some extension studies inserted more complex functions into the observer structure instead of the correcting factor $ke$ to compensate for normal force, temperature or other environmental variations \cite{canudas_IFAC.1996,canudas_IEEEICDC.2003}.

First of all, (\ref{Canudas_z}) and (\ref{Canudas_u}) introduce an observer-controller structure. Hence, the proposed scheme is not designed for arbitrary input and, therefore, cannot be utilized as a standalone observer in a straightforward way. Secondly,  the SPR condition is a conservative constraint on the controller structure because of the following reasons:
\begin{enumerate}
	\item Some common controller types such as PIDs with pure integral operation, do not satisfy the SPR condition. To get around this problem, one can utilize PID with filtered integral action with a cut-off frequency $\tau$ such that $C(s)=K_{p}+K_{d}s+\frac{K_{i}}{(\tau s+1)}$ \cite{canudas_IFAC.1996}.	
	\item A feed-forward path generally enhances the performance of the system. However, although $T(s)$ requires a mandatory feed-forward term, one may not desire to employ such a path for some reasons.  
\end{enumerate}  
Related to the first reason, \cite{canudas_IJACSP.1997} argue that a PID controller without a filtering operation can fulfill the SPR requirement within a limited frequency interval. Surprisingly, this study experimentally indicates that the closed-loop system remains stable even out of this frequency range. Regardless, mathematically speaking, this is a gap in the stability analysis that deserves further research and should be avoided to guarantee the robustness and stability of the system. Consequently, this scheme requires a control input, $u$, in a particular form since it addresses estimation and control objectives together. Distinctively, our purpose is to develop a friction estimation scheme for an arbitrary control input suitable for complex feedback strategies such as nested loops.

We note that \cite{Canudas_Kelly_AJC.2007} utilized an approach similar to \cite{Canudas_IEEETAC.1995}. In the latter, the authors considered the non-adaptive version of the control methodology introduced in \cite{Slotine_IEEEToAC.1988} and extended it to similar mechanical systems with LuGre friction. An error term similar to the one given in (\ref{Canudas_z}), which contains both position and velocity terms, is utilized in the proposed observer structure. We note that the observer structures given in these references are not standalone observers, i.e. they are not designed to estimate the friction for arbitrary input, but they are a part of the proposed observer based controller structure where the controller to be utilized to produce the necessary input to achieve stabilization is designed by considering the passivity of the overall structure. On the other hand, our proposed observer structure is aimed at the friction estimation independent of the control input. Although both structures employ an error term $e$, meaning of the error is also different. While in the \cite{Canudas_IEEETAC.1995}, the error term is the difference between actual and desired velocities, which is given; in our case, the error is the difference between the actual and estimated velocities, which is part of the observer state and hence is not known a priori. In the sequel, we refer to the observer structure given by (\ref{Canudas_z})-(\ref{Canudas_u}) as the existing observer.
\subsection{Proposed Observer}
To improve the observer given by (\ref{z_hat_dot}), let us propose the following observer structure:
\begin{eqnarray}
\label{obs_1}
\dot{\hat{z}}&=&w-\sigma_{0}\frac{|w|}{h(w)}\hat{z}+K_{1}(w-\hat{w}),\\
\label{obs_2}
J\dot{\hat{w}}&=&-\sigma_{0}\hat{z}-\sigma_{1}\dot{\hat{z}}-F_{v}w+u+K_{2}(w-\hat{w}),
\end{eqnarray}
where $K_{1}$ and $K_{2}$ are observer parameters to be determined. $\hat{w}$ is the estimation of $w$, and $u$ is the control input used in (\ref{mechanical_sys1}). The first three terms in the right hand side of (\ref{obs_2}) are the friction estimation given by (\ref{F_est}). Note that here, we assume that $w$ is available from measurement, yet we propose an estimation for it, and the basic rationale behind this might not be apparent at this point. We will clarify the importance of (\ref{obs_2}) in the stability analysis in the sequel.

We have already defined the state estimation error, $e_{z}$. Likewise, we can define the velocity estimation error as $e_{w}=w-\hat{w}$. Then, by using (\ref{mechanical_sys1}), (\ref{LuGre1})-(\ref{LuGre2}), and (\ref{obs_1})-(\ref{obs_2});
\begin{eqnarray}
\label{err_z_dot_new}
\dot{e}_{z}&=&-\sigma_{0}\frac{|w|}{h(w)}e_{z}-K_{1}e_{w},\\
\label{err_w_dot_new}
\dot{e}_{w}&=& -\frac{\sigma_{0}}{J}e_{z}-\frac{\sigma_{1}}{J}\dot{e}_{z}-\frac{K_{2}}{J}e_{w}\nonumber \\
&=&-\frac{\sigma_{0}}{J}e_{z}+\frac{\sigma_{0}\sigma_{1}}{J}\frac{|w|}{h(w)}e_{z}-K_{3}e_{w}.
\end{eqnarray}
where $K_{3}=(K_{2}-\sigma_{1}K_{1})/J.$
In the sequel, we will show that by choosing suitable observer parameters, we may obtain various stability properties for the observer given by (\ref{obs_1})-(\ref{obs_2}). To elaborate further, let us choose the following Lyapunov function such $V=Ae_{z}^{2}+Ce_{w}^{2}$ where $A>0$, $C>0$ are constant real numbers.Then, we obtain
\begin{eqnarray}
\label{V_dot_new}
\dot{V} &=&2Ae_{z}\dot{e}_{z}+2Ce_{w}\dot{e}_{w} \nonumber \\ 
&=&-2A\sigma_{0}\frac{|w|}{h(w)}e_{z}^{2}-2CK_{3}e_{w}^{2}\nonumber \\ 
&+&\left( \frac{2C\sigma_{0}\sigma_{1}}{J}\frac{|w|}{h(w)}-\frac{2C\sigma_{0}}{J}-2AK_{1}\right)e_{w}e_{z}. 
\end{eqnarray}
Now for any $A>0$, $C>0$, let us choose $K_{1}$ and $K_{2}$ as
\begin{equation}
\label{K1}
K_{1} =\frac{C\sigma_{0}}{AJ}\left(\sigma_{1}\frac{|w|}{h(w)}-1\right),~~~~K_{2}=\alpha+\sigma_{1}K_{1}
\end{equation}
where $\alpha>0$ is an arbitrary constant. Hence, we obtain $K_{3}=\alpha/J>0$.
\begin{lemma}\label{lemma4}
	Consider the system in (\ref{mechanical_sys1}), the friction model given by (\ref{LuGre1})-(\ref{Lugre3}), and the observer defined as (\ref{obs_1})-(\ref{obs_2}) with the parameters chosen as in (\ref{K1}),
	\begin{enumerate}
		\item If $w\in \mathbf{L_{\infty}}$, then $e_{w}\to 0$ as $t\to \infty$.\\
		\item If $w$ satisfies the conditions given in Lemma \ref{lemma1} or Lemma \ref{lemma2}, then the observer error dynamics given by (\ref{err_z_dot_new})-(\ref{err_w_dot_new}) are exponentially stable.
	\end{enumerate}	
\end{lemma}
\begin{proof}
	From (\ref{V_dot_new})-(\ref{K1}), it follows that $\dot{V}=-2A\sigma_{0}\frac{|w|}{h(w)}e_{z}^{2}-2C\frac{\alpha}{J}e_{w}^{2}.$ Hence, $\dot{V}\leq 0$ and we have $e_{z}\in \mathbf{L_{\infty}}$, $e_{w}\in \mathbf{L_{\infty}}$. Also, by integrating $\dot{V}$, we obtain $2C\frac{\alpha}{J}\int_{0}^{t}e_{w}^{2}dt \leq V(0).$ Hence, $e_{w} \in \mathbf{L_{2}}$. If $w\in \mathbf{L_{\infty}}$, from (\ref{err_w_dot_new}), it follows that $\dot{e}_{w}\in \mathbf{L_{\infty}}$ as well. Hence, by Barbalat's Lemma, it follows that $e_{w}\to 0$ as $t\to \infty$. The proof of $(ii)$ easily follows from the proof of Lemma \ref{lemma1} or Lemma \ref{lemma2}. 
\end{proof}
At this point, the advantage of the observer given by (\ref{obs_1})-(\ref{obs_2}) over the one given in (\ref{z_hat_dot}) is unclear. To see the difference between these observers, consider the case where $w(t)\equiv0$. In the existing design, (\ref{z_hat_dot}) implies $\dot{e}_z=0$, hence $e_{z}(t)=e_{z}(0) \in \mathbf{L_{\infty}}$. However, in the proposed observer case given by (\ref{err_z_dot_new})-(\ref{err_w_dot_new}), (\ref{K1}) implies $K_{1}=-C\sigma_{0}/AJ$ and $K_{3}=\alpha/J>0$. Therefore, it becomes exponentially stable. This simple observation can be extended to the case where $w\in \mathbf{L_{\infty}}$, which will be exploited below.

Assume that $w\in \mathbf{L_{\infty}}$, and we have $|w(t)|\leq M$ for some $M>0.$ If $K_{1}$ and $K_{2}$ are chosen as in (\ref{K1}), by using (\ref{proof1_eq}), we obtain:
\begin{align}
\label{eq1}
\sigma_{1}K_{1}&=\sigma_{1}\frac{C\sigma_{0}}{AJ}\left(\frac{|w|}{h(w)}-1 \right) \leq \sigma_{1}\frac{C\sigma_{0}}{AJ}\left(\frac{M}{C_{1}}-1 \right),\\
\label{eq2}
K_{2}&=\alpha+\frac{C}{AJ}\sigma_{0}\sigma_{1}\left(\frac{M}{C_{1}}-1 \right).
\end{align}
By choosing $\alpha>0$, we have $K_{2}>0$ and $JK_{3}=K_{2}-\sigma_{1}K_{1}\geq \alpha$ as well.
\begin{lemma}\label{lemma5}
	Consider the system given by (\ref{mechanical_sys1}), the friction model given by (\ref{LuGre1})-(\ref{Lugre3}), and the observer given by (\ref{obs_1})-(\ref{obs_2}). Assume that $w\in \mathbf{L_{\infty}}$ and the observer parameters are chosen as in (\ref{K1}) and (\ref{eq2}). Then, $e_{w}$, $e_{z}\in \mathbf{L_{2}}$. Moreover, the error dynamics given by (\ref{err_z_dot_new})-(\ref{err_w_dot_new}) is asymptotically stable, i.e. $e_{w}\to 0$ and $e_{z}\to 0$ as $t\to \infty$.
\end{lemma}
\begin{proof}
	We now have	$\dot{V} \leq -2A\sigma_{0}\frac{|w|}{h(w)}e_{z}^{2}-2C\frac{\alpha}{J} e_{w}^{2}.$	Similar to Lemma \ref{lemma4}, we have $e_{w}$, $\dot{e}_{w} \in \mathbf{L_{\infty}}$ and $e_{w}\in \mathbf{L_{2}}$; hence, by Barbalat's lemma, we have $e_{w}\to 0$ as $t\to \infty$. Now consider the error equation given by (\ref{err_w_dot_new}). By taking the Laplace transform, we can rewrite this equation as
	\begin{equation}
	\label{err_z_lap}
	E_{z}(s)=-\frac{Js+K_{2}}{\sigma_{1}s+\sigma_{0}}E_{w}(s)=-G(s)E_{w}(s),
	\end{equation}
	where $E_{w}$ and $E_{z}$ refers to the Laplace transforms of $e_{w}(.)$ and $e_{z}(.)$ respectively. Hence, we can view (\ref{err_w_dot_new}) as an LTI system with input $e_{w}$, output $e_{z}$ and a transfer function $G(s)$ given by (\ref{err_z_lap}). Note that $G(s)$ is stable and $|G(j2\pi f)|<C$ for some $C>0$. Hence, it is $\mathbf{ L_{2}}$ stable. By Parsavel's Theorem, it then follows that $e_{z}\in \mathbf{ L_{2}}$ as well, see e.g. Theorem 5.4 in \cite{khalil_nonlinear.2002}. Since $e_{z}\in \mathbf{ L_{\infty}}$, it follows from Barbalat's Lemma that $e_{z}(t)\to 0$ as well. 
\end{proof}
\begin{remark}\label{remark4}
	Lemma \ref{lemma5} shows the primary difference between the observer given by (\ref{z_hat_dot}) and our proposed observer given by (\ref{obs_1})-(\ref{obs_2}). If $K_{2}>0$ is a constant and $e_{w}\in \mathbf{ L_{2}}$, the interpretation of (\ref{err_w_dot_new}) as an LTI system given by (\ref{err_z_lap}) enables us to conclude that $e_{z}\in \mathbf{ L_{2}}$ as well. This also justifies the utilization of a velocity estimation despite it is measurable.  Note that this result is independent of the input $u(.)$ used in (\ref{mechanical_sys1}) as long as $w\in \mathbf{ L_{\infty}}$. 
\end{remark}
Lemma \ref{lemma5} gives us a result of the asymptotic stability of the observer given by (\ref{obs_1})-(\ref{obs_2}). Next, we will show that we can achieve exponential stability by choosing appropriate observer parameters. Now consider the following Lyapunov function $V=Ae_{z}^{2}+Be_{z}e_{w}+Ce_{w}^2$ for the system given by (\ref{err_z_dot_new})-(\ref{err_w_dot_new}) where $A$, $B$ and $C$ are some constants to be determined later. Clearly, $V$ is a positive function if $A>0,~C>0,~B^{2}-4AC<0.$ Using (\ref{err_z_dot_new}), (\ref{err_w_dot_new}) and derivative of $V$, we obtain:
\begin{eqnarray}
\label{new_V_dot}
\dot{V}&=&-\sigma_{0}\left( 2A-\frac{B\sigma_{1}}{J}\right) \frac{|w|}{h(w)}e_{z}^{2}-\frac{B\sigma_{0}}{J}e_{z}^{2}\nonumber \\ &-&(BK_{1}+2CK_{3})e_{w}^{2}+\sigma_{0}\left(\frac{2C\sigma_{1}}{J}-B\right)\frac{|w|}{h(w)}e_{z}e_{w}\nonumber \\
&-&\left(BK_{3}+\frac{2C\sigma_{0}}{J}+2AK_{1}\right) e_{z}e_{w}.
\end{eqnarray}
Now, we choose $A$, $B$, $C$, $K_{1}$ and $K_{3}$ so that the following equations are satisfied and V is positive.
\begin{eqnarray}
2A\geq \frac{B\sigma_{1}}{J}>0, \label{ineq1}\\
BK_{1}+2CK_{3}>0\label{ineq2},\\
\frac{2C\sigma_{1}}{J}-B=0,\label{ineq3}\\
BK_{3}+\frac{2C\sigma_{0}}{J}+2AK_{1}B=0.\label{ineq4}
\end{eqnarray}
By straightforward calculations, it can be shown that $V$ and (\ref{ineq1})-(\ref{ineq2}) have solutions. In the sequel, we will show a step-by-step solution.
\begin{itemize}
	\item Choose any $C>0$, $\alpha>0$, $\beta>0$.\\
	\item Choose $B$ and $A$ such that
	\begin{equation}
	B=\frac{2C\sigma_{1}}{J}>0,~~~~~~ A= \frac{B\sigma_{1}}{2J}+\frac{\beta}{2} \label{AB}
	\end{equation}
	\item Choose $K_{1}$ and $K_{3}$ as
	\begin{equation}
	K_{1}=- \frac{2C(\sigma_{1}\alpha+\sigma_{0})}{\beta J},~~~~K_{3}=-\frac{\sigma_{1}K_{1}}{J}+\alpha \label{K3_new}.
	\end{equation}
\end{itemize}
Therefore, we have $K_{2}= JK_{3}+\sigma_{1}K_{1}=J\alpha$. Note that with this selection we have $B^{2}-4AC= -2\beta C \label{eq_1}.$ Hence, $V$ becomes a positive definite function. Likewise, $ BK_{1}+2CK_{3}=2C\alpha>0$ and $BK_{3}+\frac{2C\sigma_{0}}{J}+2AK_{1}=\frac{2C}{J}(\sigma_{1}\alpha+\sigma_{0})+\beta K_{1}=0$.
\begin{lemma}\label{lemma6}
	Consider the system given in (\ref{mechanical_sys1}), the friction model given by (\ref{LuGre1})-(\ref{Lugre3}), and the observer given by (\ref{obs_1})-(\ref{obs_2}). Let the observer parameters be chosen as in the step-by-step solution where $\alpha>0$ and $\beta>0$ are arbitrary constants. Then, the error dynamics given by (\ref{err_z_dot_new})-(\ref{err_w_dot_new}) is exponentially stable.
\end{lemma}
\begin{proof}
	Let $C>0$, $A$ and $B$ as given by (\ref{AB}) and consider the aforementioned positive definite Lyapunov function $V$. From (\ref{new_V_dot}), we obtain $\dot{V}=-\sigma_{0}\beta\frac{|w|}{h(w)}e_{z}^{2}-\frac{B\sigma_{0}}{J}e_{z}^{2}-2C\alpha e_{w}^{2}$ which proves that $-\dot{V}$ is positive definite as well. Since $V$ and $-\dot{V}$ are quadratic, exponential stability follows from standard Lyapunov Stability Theorems, e.g. Theorem 4.10 in \cite{khalil_nonlinear.2002}. 
\end{proof}

\section{Simulations}
We conduct simulations for system parameters provided in Table~\ref{fric_params} \cite{canudas_IFAC.1996,canudas_IJACSP.1997}. First, we consider a velocity control system, as in Fig. \ref{sys}(a), without any feed-forward compensator. For $K_{p}=1.6$ and $K_{i}=0.16$, we obtain a PI controller with around $115$ Hz bandwidth at $-3$ dB.

\begin{table}[t]
	\caption{Simulation parameters}
	\label{fric_params}
	\vspace{2mm}
	\centerline{
		\begin{tabular}{lccc}
			\hline\hline
			Parameter & Notation & value & Unit \\
			\hline\hline
			Stribeck velocity & $v_{s}$& $0.01$ &  $rad/s$\\
			Stiffness coefficient&$\sigma_{0}$& $260$ & $N.m/rad$\\
			Damping coefficient&$\sigma_{1}$& $0.6$ & $N.m.s/rad$\\	
			Coulomb friction& $F_{c}$&$0.285$ & $N.m$\\ 
			Stick friction  & $F_{s}$&$0.335$ & $N.m$\\
			Viscous friction  & $F_{v}$&$0.018$ & $N.m.s/rad $\\
			Total inertia		&$J$&$0.0022$ & $kg.m^{2}$\\ 
			\hline\hline
	\end{tabular}}
\end{table}

\begin{figure}[b]
	\centering
	\includegraphics[width=1\columnwidth]{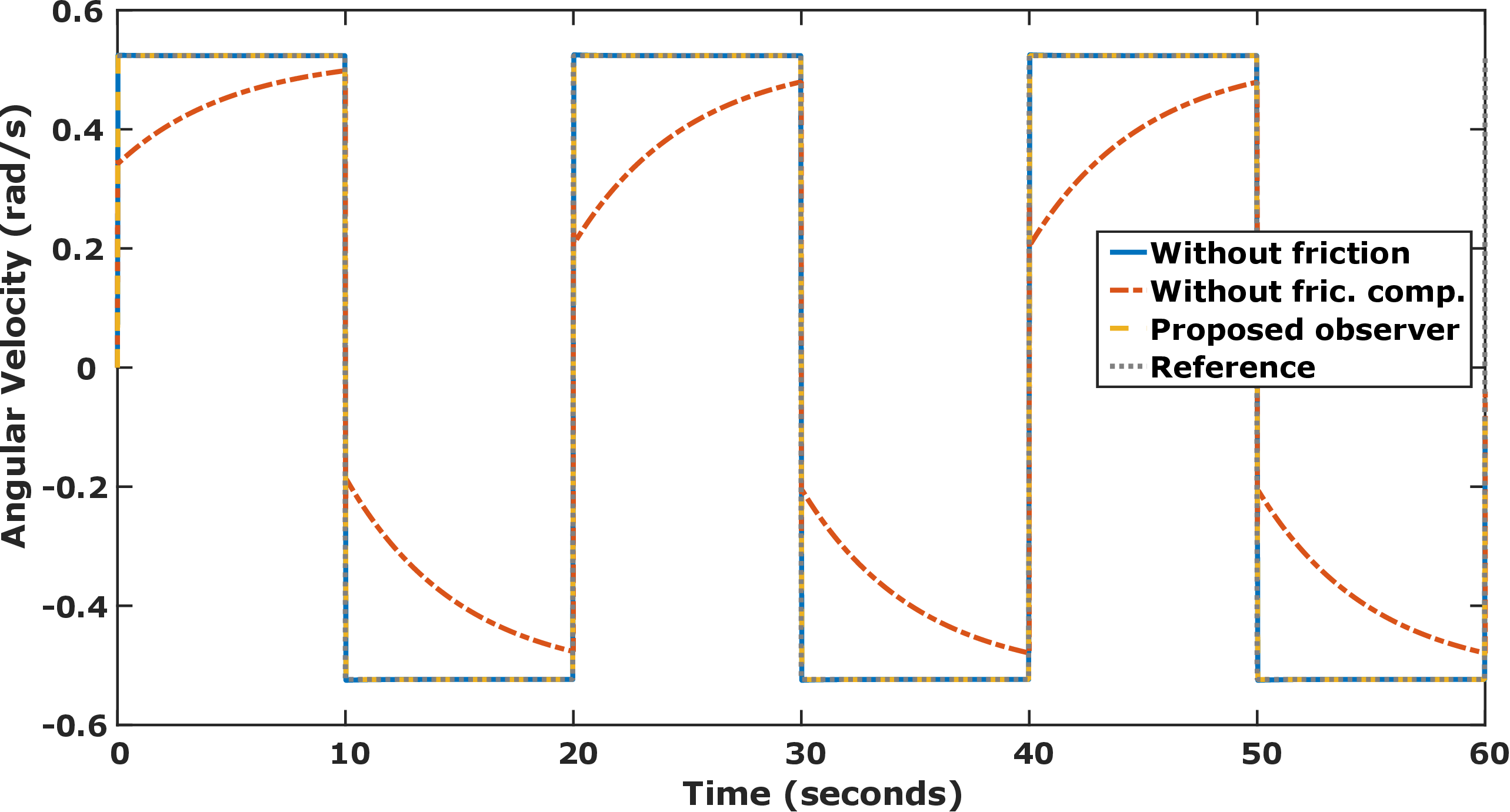}
	\caption{\label{vel_response} Velocity tracking response of the system with the proposed observer under LuGre Friction.}
\end{figure}

As it can be observed in Fig. \ref{vel_response}, although the PI controller provides a promising velocity tracking response when there is no friction, it exhibits severe performance degradation under the presence of LuGre Friction. Clearly, the proposed observer scheme is a remedy to compensate for the tracking performance due to LuGre friction. Indeed, a large integral gain improves the disturbance rejection performance, especially at lower frequencies. In this sense, the integral coefficient can be increased by aiming at friction compensations. On the other hand, although the integral part eliminates the steady state, it may also introduce significant phase lag and overshoot. Moreover, the integral gain may not be apparent in the command response until it is excessive to cause peaking and instability \cite{ellis_book.2012}. Hence, in such cases, estimating the disturbance with an observer is a widespread approach to enhance the disturbance performance. Motivated by this fact, we combine our friction observer structure with a moderate velocity controller to demonstrate the technique's efficacy. 
In these simulations, we designate $K_{1}=-10.24$ in (\ref{obs_1}) and $K_{2}=22$ in (\ref{obs_2}) while $k=0.35$ in (\ref{Canudas_z}). 


\begin{figure}[b]
	\centering
	\includegraphics[width=1\columnwidth]{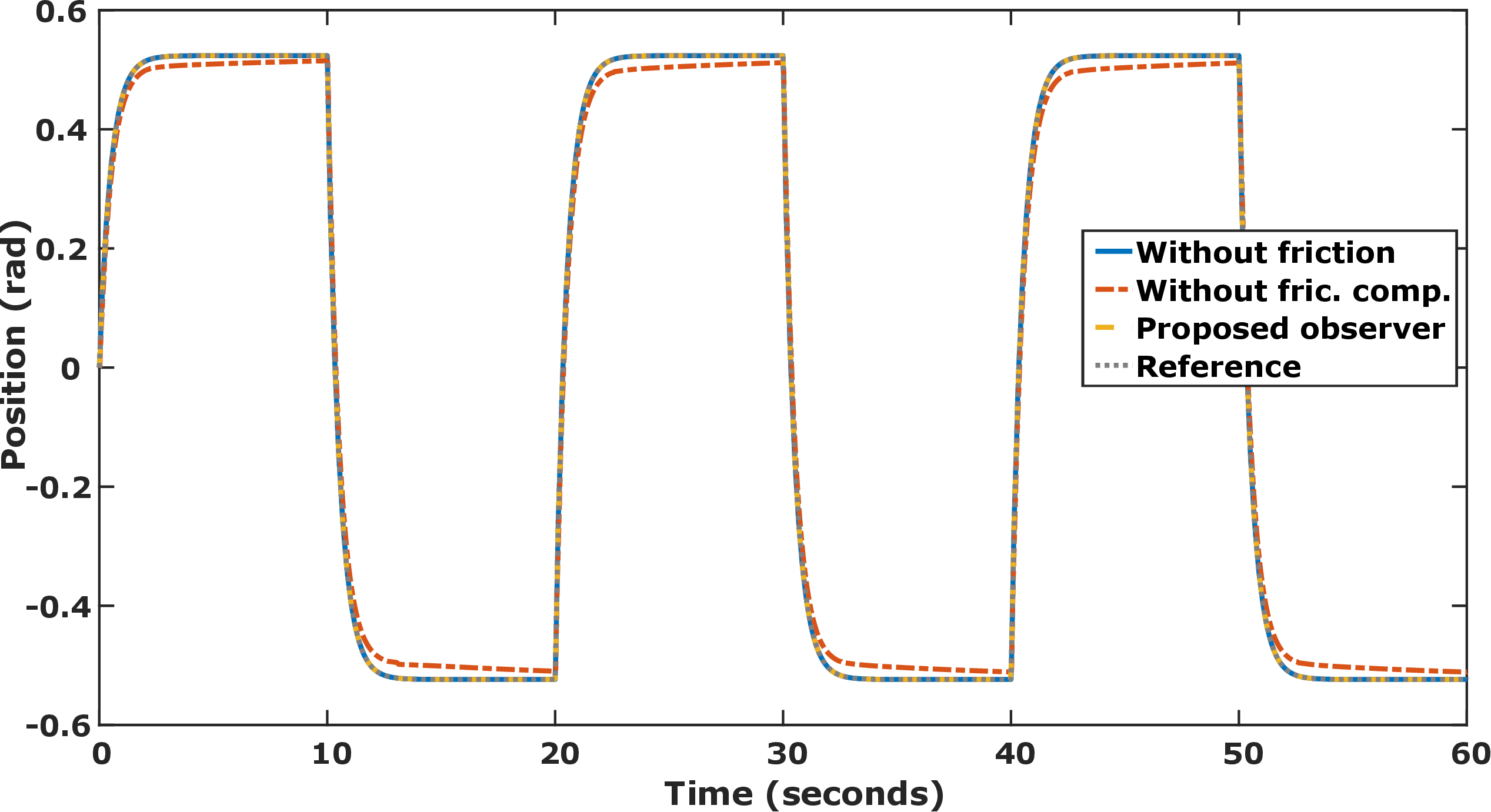}
	\caption{\label{pos_res}Position tracking response of the system with the proposed observer under LuGre friction.}
\end{figure}

Then, we also consider a position control system, as in Fig. \ref{sys}(b). For the conducted position response simulations, we determine the coefficients of the PI controller as $K_{p}=15$ and $K_{i}=1.55$. Thus, approximately $20$ Hz controller bandwidth is achieved. In a position-tracking scenario, the output signal exhibits an oscillatory transient response due to the integral term for aggressive reference inputs. Therefore, a pre-filter shown in Fig. \ref{sys}(b) is generally utilized to smoothen the reference signal. Hence, we adopt a low-pass filter, $H(s) = 2/(s+2)$. In Fig. \ref{pos_res}, it can be distinguished that the proposed observer structure also performs well in a position-tracking objective. 



\section{Conclusion}
The LuGre friction model is vastly used in the literature for friction modeling and compensation. Motivated by this fact, we propose a new observer design based on the LuGre friction model and provide the necessary conditions for the asymptotic stability of our observer by using Lyapunov functions. Uniquely, we used a velocity error in this design to enhance the friction estimation performance. Our simulations demonstrate that the introduced observer structure improves the performance of the velocity and position control systems in the presence of LuGre friction. Moreover, the developed observer is a stand-alone design compared with a benchmark friction observer. Therefore, the main contribution of this paper is to provide an LuGre friction observer structure that can be employed with any controller for friction compensation in mechanical systems.

\bibliographystyle{IEEEtran}
\bibliography{lcsys}   

\end{document}